\documentclass[12pt,reqno]{amsart}
\usepackage{amsfonts,amstext}
\usepackage{tikz}
\usepackage{amsmath,amstext,amssymb,amscd}
\usepackage{multicol}
\usepackage{calligra}
\usepackage{amsfonts}
\usepackage{amsthm}
\usepackage{amsmath}
\usepackage{amscd}
\usepackage[latin2]{inputenc}
\usepackage{t1enc}
\usepackage[mathscr]{eucal}
\usepackage{indentfirst}
\usepackage{graphicx}
\usepackage{graphics}
\usepackage{pict2e}
\usepackage{epic}
\numberwithin{equation}{section}
\usepackage[margin=2.9cm]{geometry}
\usepackage{epstopdf}

\newtheorem{rema}{Remark}
\newtheorem{prop}{Proposition}
\newtheorem{conj}{Conjecture}
\newtheorem{theorem}{Theorem}
\newtheorem{lemma}[theorem]{Lemma}

\newtheorem{example}{Example}
\newtheorem{algorithm}{Algorithm}

\DeclareMathAlphabet{\mathcalligra}{T1}{calligra}{m}{n}

\title [Inner Product in Highest-Weight Representation]
{Inner Product in Highest-Weight Representation}
\author{Chuanzhong Li$^1$}
\author{Zhisheng Liu$^2$}
\author{Bao Shou$^3$}

\address{$^1$ College of Mathematics and Systems Science\\ Shandong University of Science and Technology\\ Qingdao,266590, P.R.China}

\address{$^2$Institute  of Theoretical  Physics\\
Chinese Academy  of Sciences\\
Beijing, 100190, China}

\address{$^3$Center  of Mathematical  Sciences,
Zhejiang University,
Hangzhou, 310027, China}

\email{lichuanzhong@sdust.edu.cn, zsliu@itp.ac.cn, bsoul@zju.edu.cn}

\subjclass[2010]{16Z05}

\keywords{ Highest-weight representation, inner product, iterative algorithm, weight, Kapustin-Witten  equations}

\date{}                                         % Activate to display a given date or no date

\begin{document}

\begin{abstract}
  In this paper, we study the inner product of states corresponding to weights  of finite-dimensional  highest-weight representations of classical groups.
  We prove that the  action of the  raising operators would reduce a state of hight-weight representation    to a linear combination of states of highest-weight representation,   with the level decreased by one.
  Then we propose an iterative  algorithm for calculating the inner products of sates efficiently,  revealing  the intricate structure of the representation.

  As applications, we discuss   the  unitarity of the highest-weight representation and  propose a conjecture. We determine the norm of a special class of states.   And we  completely determine the inner products of states of the minuscule representations. The algorithm proposed is   applicable to the highest-weight representation of  affine Lie algebra without modifications. These findings can be used to study the construction of  solutions to  Kapustin-Witten equations which are based on the fundamental solutions of Toda systems.

\end{abstract}
\maketitle

\tableofcontents

\section{Introduction}
Highest-weight representations are familiar to physicists. The representation of the Virasoro algebra  can be used to characterize simple conformal field  theories,  such as minimal modes. The other states of representation can be obtained by successive applications of Virasoro operators on the highest-weight states.  The inner product of states   can be used to characterize  the  unitarity of representations,  revealing  the intricate structure of conformal theories.

 In this paper, we study  the inner product of states of highest-weight representations of simple Lie algebra $\mathfrak{g}$ and its applications. Any finite-dimensional irreducible representation has a unique highest-weight state  $|\Lambda\rangle$.  Starting from the highest-weight state,   all the states in the representation  space  can be obtained by the action of the lowering operators of $\mathfrak{g}$  as follows
$$|\lambda\rangle=E^{-\beta} E^{-\gamma}\cdots E^{-\eta}|\Lambda\rangle \quad\textrm{for} \quad \beta,\gamma,\cdots,\eta\in \Delta_+.$$
The inner product between the state $|\lambda\rangle$ and the conjugate state  $\langle^{'}\lambda|$ of another state $|\lambda^{'}\rangle$ is
\begin{equation}\label{in}
\langle^{'}\lambda|\lambda\rangle=\langle\Lambda^{'}| E^{+\eta^{'}}\cdots E^{+\gamma^{'}} E^{+\beta^{'}}|E^{-\beta} E^{-\gamma}\cdots E^{-\eta}|\Lambda\rangle.
\end{equation}
For any state  $|\lambda\rangle$, the norm of $|\lambda\rangle$ is positive definite \cite{cft}
$$\langle\lambda|\lambda\rangle=\langle\Lambda| E^{+\eta}\cdots E^{+\gamma} E^{+\beta}|E^{-\beta} E^{-\gamma}\cdots E^{-\eta}|\Lambda\rangle >0.$$

The inner product (\ref{in}) appears   in the solutions of  the Kapustin-Witten (\textbf{KW}) equations.
The following inner product is one of the  factors     conjectured in  \cite{vm}
\begin{equation}\label{inpro}
  \langle\upsilon^{i}_{w}|\upsilon^{j}_{w}\rangle=\langle\Lambda_s|E^{+}_{i_1} \cdots E^{+}_{i_n}|E^{-}_{j_n}\cdots E^{-}_{j_1}|\Lambda_s\rangle,
\end{equation}
where $E^{-}_{j_n}\cdots E^{-}_{j_1}$ and $E^{+}_{i_1} \cdots E^{+}_{i_n}$  are two  sequences from the highest-weight $\Lambda_s$ to the  weight $w$.  The  superscript $i$ and $j$ denote the $i$th and $j$th  paths, respectively.
The calculation of the commutation of operators in equation (\ref{inpro}) is tedious and must be done manually for a large number of operators.
In addition, to define the vector $|\upsilon_{w}(\hat{\omega})\rangle$, it need  to consider all the paths   from the highest-weight state to $w$. Each branch node in the diagram increases  the number of paths,  such as the node $(1,-1,1)$ shown in Fig.(\ref{A3-010}).
 As the rank of $\mathfrak{g}$ increases, the number of paths, as well as the number of weights, grows rapidly. More details will be provided in Section \ref{kw}.

\begin{figure}[!ht]
  \begin{center}
    \includegraphics[width=4in, bb=10 10 500 300]{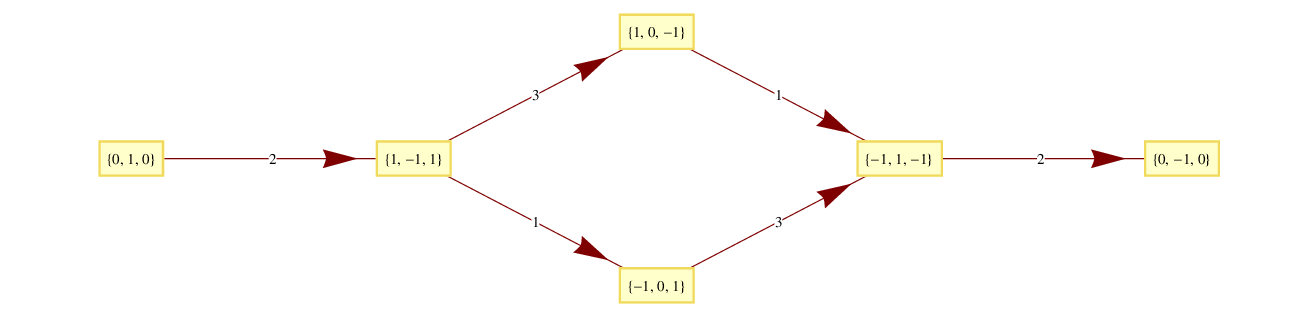}
  \end{center}
  \caption{Weights  in the fundament representation $\rho_2$ of $A_3$.}
  \label{A3-010}
\end{figure}
We illustrate these problems through an example.
 \begin{example}
 Fundamental  representation $\rho_2$ for $A_3$: Let $\hat{\omega}$ be the weight vector  ended.  As shown in Fig.(\ref{A3-010}), there are four paths from the highest-weight $(0,1,0)$ to $(-1,1,-1)$.
To construct the norm of the  state $|\upsilon_{(-1,1,-1)}(\hat{\omega})\rangle$,  we need to compute  the  following two norms
\begin{eqnarray*}
  \langle\upsilon^{1}_{(-1,1,-1)}|\upsilon^{1}_{(-1,1,-1)}\rangle=\langle\Lambda_2|E^{+}_{2} E^{+}_{1}E^{+}_{3}|E^{-}_{3} E^{-}_{1}E^{-}_{2}|\Lambda_2\rangle, \\
  \langle\upsilon^{2}_{(-1,1,-1)}|\upsilon^{2}_{(-1,1,-1)}\rangle=\langle\Lambda_2|E^{+}_{2} E^{+}_{3}E^{+}_{1}|E^{-}_{1} E^{-}_{3}E^{-}_{2}|\Lambda_2\rangle,
\end{eqnarray*}
which correspond to two sequences from the highest-weight $(0,1,0)$ to $(-1,1,-1)$.
We also  need  to compute the  following inner products,
\begin{eqnarray*}
  \langle\upsilon^{1}_{(-1,1,-1)}|\upsilon^{2}_{(-1,1,-1)}\rangle=\langle\Lambda_2|E^{+}_{2} E^{+}_{1}E^{+}_{3}|E^{-}_{1} E^{-}_{3}E^{-}_{2}|\Lambda_2\rangle, \\
 \langle\upsilon^{2}_{(-1,1,-1)}|\upsilon^{1}_{(-1,1,-1)}\rangle=\langle\Lambda_2|E^{+}_{2} E^{+}_{3}E^{+}_{1}|E^{-}_{3} E^{-}_{1}E^{-}_{2}|\Lambda_2\rangle,
\end{eqnarray*}
 where the states and their conjugate states   are constructed by   different paths.
 \end{example}

In this paper, we  focus on the calculation of the inner product (\ref{inpro}). %Our motivation is that a more thorough understanding of the construction of solutions to  the \textbf{KW} equations might lead to progress.
In Section 2, we introduce the highest-weight representations as a  preparation.
In Section 3,  we propose an algorithm for calculating the inner product of states of the highest-weight representations,  corresponding to a path in the weight diagram.
The algorithm is an iterative process based on Theorem \ref{ite}, which significantly reduces the computational workload.
It  can be used to calculate inner product  of any irreducible finite-dimensional representations of classical groups in addition to the fundamental representations appearing in the \textbf{KW} equations.

In Section 4, we introduce applications of the algorithm.
We discuss the unitarity of the highest-weight representations and propose one conjecture.  We calculate the norms of  a special kind of states. And we also determine the inner product of states in the minuscule representations.  Finally, the application to  the solution of \textbf{KW} equations is pointed  out.
In the appendix, we give an example to illustrate the proposed algorithm  in detail. By using the proposed algorithm, the inner product of a state can be recorded in less than one page. In contrast, without the algorithm,   it    would take twenty-five pages to record the calculation process.

\section{Preliminary on Highest-Weight Representations}
In the Chevalley basis of a Lie algebra $\mathfrak{g}$, the corresponding raising and lowering operators are denoted by $E_i^{\pm}$
  for a simple root $\alpha_i$, and the operators corresponding to the coroots are denoted by $H_i$. The commutation relations of these operators  are given by
\begin{equation}\label{cn}
 [E_i^{+}, E_{j}^{-}] = \delta_{ji}H_{j}, \quad  [H_{i}, E^{\pm}_{j}] = \pm A_{j,i}E^{\pm}_{j}, \quad [H_i, H_j] =0 .
\end{equation}

The  finite-dimensional irreducible representation $L_{\Lambda}$ has a unique highest-weight state $|\Lambda\rangle$,  satisfying
$$E^{\alpha}|\Lambda\rangle=0,$$
for any positive root $\alpha$,  which is a linear combination of simple roots $\alpha_i$ with positive coefficients.
All states $|\lambda\rangle$ associated to the weight $\lambda$  in the representation space $L_{\Lambda}$ can be obtained by the action
$$|\lambda\rangle=E^{-\beta} E^{-\gamma}\cdots E^{-\eta}|\Lambda\rangle \quad\textrm{for} \quad \beta,\gamma,\cdots,\eta\in \Delta_+,$$
where $\Delta_+$ is the set of positive roots.
The states $|\lambda\rangle$ satisfy the following identity
\begin{equation}\label{ef}
 H_i|\lambda\rangle=\lambda_i|\lambda\rangle.
\end{equation}
The set of eigenvalues of all states in $L_{\Lambda}$ is the weight system $\Omega_{\Lambda}$.
Any weight $\lambda^{'}$ in the set $\Omega_{\Lambda}$ is such that $\lambda-\lambda^{'} \in \Delta_{\Lambda}$.  Consequently, $\lambda$ is necessarily of the form $\Lambda-\sum n_i \alpha_i $, with $n_i\in \mathbb{Z}_+$.
We call $\sum n_i $ the level of the weight $\lambda$ in the representation $\Lambda$.
These weights  can be obtained  by the action of the lowering operators $E_{*}^-$ of $\mathfrak{g}$ as follows
$$|\lambda\rangle=E_{j_{n(w)}}^-\cdots E_{j_1}^-|\Lambda\rangle,$$
whose conjugate state is defined by
 $$\langle\lambda|=\langle\Lambda | E_{j_{1}}^+\cdots E_{j_{n(w)}}^+ |.$$
Then we have
\begin{eqnarray}\label{com}
    H_i|\lambda\rangle &=&H_iE_{j_{n(w)}}^-\cdots E_{j_1}^-|\Lambda\rangle\nonumber\\
   &=& (E_{j_{n(w)}}^-H_i-E_{j_{n(w)}}^-A_{j_{n(w)},i})\cdots E_{j_1}^-|\Lambda\rangle\nonumber \\
   &=&-\sum_{b=1}^{n(w)} E_{j_{n(w)}}^-\cdots A_{j_b, i} E_{j_1}^-\cdots E_{j_1}^-|\Lambda\rangle+E_{j_{n(w)}}^-\cdots E_{j_1}^-H_i|\Lambda\rangle\nonumber\\
   &=& (\Lambda_i-\sum_{b=1}^{n(w)}A_{j_b, i}) E_{j_{n(w)}}^-\cdots E_{j_1}^-|\Lambda\rangle.
\end{eqnarray}
According to formula (\ref{ef}), we have
\begin{equation}\label{ieq}
  \lambda_i=\Lambda_i-\sum_{b=1}^{n(w)}A_{j_b, i}.
\end{equation}

There are  two positive integers $p_i$ and $q_i$, such that
\begin{eqnarray}\label{pq}
  (E^{+}_{i})^{p_i+1}|\lambda\rangle &\sim & E^{+}_i|\lambda+p_i\alpha_i\rangle=0,\\
  (E^{-}_i)^{q_i+1}|\lambda\rangle &\sim & E^{-}_i|\lambda-q_i\alpha_i\rangle=0,
\end{eqnarray}
for any simple root $\alpha_i$. The integers $p_i$ and $q_i$ satisfy the following identity
$$2\frac{(\alpha_i,\lambda)}{|\alpha_i|^2}=(\alpha_i^{\vee},\lambda)=-(p_i-q_i).$$
This identity is crucial for determining all the weights in the weight system $\Omega_{\Lambda}$. By progressing level by level, we determine the value of $p_i$
  at each step. Clearly, $\lambda-\alpha_i$ is also a weight if $q_i$ is nonzero, that is, if $\lambda_i+p_i > 0$. We emphasize the following string of weights corresponding to a weight $\lambda$
\begin{equation}\label{string}
  \lambda+p_i\alpha_i,\cdots,\lambda,\cdots,\lambda-q_i\alpha_i.
\end{equation}
Note that the weight $\lambda+p_i\alpha_i$ satisfies formula (\ref{pq}), which will play a significant role in the algorithm proposed in Section 3.

The procedure for constructing  all the weights in the representation is  as follows. We start with the highest-weight $\Lambda=(\Lambda_1, \cdots,\Lambda_r)$:
\begin{enumerate}
   \item  Construct the sequence of weights $\Lambda-\alpha_i, \Lambda-2\alpha_i, \cdots, \Lambda-\Lambda_i\alpha_i$.
   \item  Repeat this process  with $\Lambda$ replaced by each of the weights just obtained.
   \item  Continue this process until no more weights with positive  Dynkin labels are produced.
 \end{enumerate}
 Simple examples will clarify the method.
 \begin{example}
     Consider the fundamental representation $\rho_2$ of $A_3$, characterized by the highest weight $(0,1,0)$, as illustrated in Fig. (\ref{A3-010}). The weights at each step can be derived using the aforementioned procedure.
 \end{example}
\begin{flushleft}
Another example is the fundamental representation $\rho_2$
  of $G_2$, whose weights can be read from Fig.(\ref{G2}) in Section 3.
\end{flushleft}

\section{Inner Product}
In this section,    we first  prove that  the inner product of the states defined by   different weights is zero.
And  then, we examine  the inner product of  states defined by  the same weight but with different paths from the highest weight to the terminating weight    in the weight diagram.

In Theorem  \ref{ite}, we prove that the  action of the  raising operators would reduce a state of hight-weight representation    to a linear combination of states of highest-weight representation,   with the level decreased by one.
Based on this crucial result, we propose an iterative algorithm   for calculating the inner product efficiently, which is the main result of this paper.
 %As an application,  we determine norms of special states. We also determine the inner product  of states in the minuscule representations.

\subsection{Inner Product of  Different States}
Let $|\upsilon^{m}_{w}\rangle$ denote a state of level $m$  along a path   in the weight diagram ended with the  weight $w$.
And  $|\upsilon^{m,n}_{w}\rangle$ denote a state of level $m$  along the $n$th path from the highest-weight vector to the weight $w$.

First, we derive a useful identity.
According to the commutation  relations (\ref{cn}), for the highest-weight $\Lambda=\Sigma_a \lambda_a \omega_a$,  we have the following  identity,
 \begin{eqnarray}\label{le0}
  E_{a}^{+}E_{j_n}^{-} E_{j_{n-1}}^{-}\cdots E_{j_1}^{-} |\Lambda\rangle &=& (\delta_{a,j_n}H_{a}+E_{j_n}^{-}E_{a}^{+} )E_{j_{n-1}}^{-}\cdots E_{j_1}^{-} |\Lambda\rangle \nonumber\\
    &=&\sum_{i=1}^{n}E_{j_n}^{-} E_{j_{n-1}}^{-}\cdots E_{j_{i+1}}^{-}\delta_{a,j_i}H_{a} \hat{E}_{j_{i}}^{-}E_{j_{i-1}}^{-}\cdots E_{j_1}^{-}  |\Lambda\rangle\\
    &=& \sum_{i=1}^{n}\delta_{a,j_i}(\lambda_a-(\sum_{i=1}^{i-1}A_{j_l,a}))E_{j_n}^{-} E_{j_{n-1}}^{-}\cdots \hat{E}_{j_{i}}^{-}E_{j_{i-1}}^{-}\cdots E_{j_1}^{-} |\Lambda\rangle,\nonumber
 \end{eqnarray}
where the hat indicates the omission of the corresponding term.  We have used formula (\ref{com}) in the last equality. We sequentially communicate  $E^{+}_{a}$   with the lowering operators  $E^{-}_{*}$  from right to left.  Finally,  the remaining operator $E^{+}_{a}$ annihilate the highest-weight state $|\Lambda\rangle$.

The following theorem imply that the level $m$ is a quantum number, which means inner product of  states with different level are zero.
\begin{theorem}\label{ii1}
 The inner product  of the states with different level is zero.
\end{theorem}
\begin{proof}
Assume the states $|\upsilon^{m}\rangle$ and $|\upsilon^{n}\rangle$ with levels $m$ and $n$   are given by
$$|\upsilon^{m}\rangle=|E^{-}_{i_m}\cdots E^{-}_{i_1}|\Lambda\rangle, \quad\quad |\upsilon^{n}\rangle=|E^{-}_{j_n}\cdots E^{-}_{j_1}|\Lambda\rangle.$$
Without loss of generality,  assume the level $m>n$. Then the  inner product of $|\upsilon^{m}\rangle$ and $|\upsilon^{n}\rangle$ is
\begin{eqnarray*}
  \langle\upsilon^{m}|\upsilon^{n}\rangle &=& \langle\Lambda|E^{+}_{i_1} \cdots E^{+}_{i_m}|E^{-}_{j_n}\cdots E^{-}_{j_1}|\Lambda\rangle \\
  &=& \langle\Lambda|E^{+}_{i_1} \cdots E^{+}_{i_{m-1}} (\delta_{i_m,j_n}H_{i_m}+ E^{-}_{j_n}E^{+}_{i_m})E^{-}_{j_{n-1}}\cdots E^{-}_{j_1}|\Lambda\rangle \\
   &=& \langle\Lambda|E^{+}_{i_1} \cdots E^{+}_{i_{m-1}}  (\sum_{k=1}^n E^{-}_{j_n}\cdots E^{-}_{j_{k+1}}\delta_{i_m,j_n}H_{i_m}E^{-}_{j_{k-1}} \cdots E^{-}_{j_1})|\Lambda\rangle.
\end{eqnarray*}
According to formula (\ref{le0}), the operator $\delta_{i_m,j_n}H_{i_m}$ can be seen as an undetermined constant $h_{i_m,j_k}$ because it is the eigenvalue of the operator  $\delta_{i_m,j_n}H_{i_m}$ acting  on the state $E^{-}_{j_{k-1}} \cdots E^{-}_{j_1}|\Lambda\rangle$.
The action of the operator $E^{+}_{*}$   decrease the number of the operators $E^{-}_{*}$ of $|\upsilon_{n}\rangle$ by one. Then there is no $E^{-}_{*}$ left after $n$ times of the   actions of the operators $E^{+}_{*}$. And the remaining  operators $E^{+}_{i_1} \cdots E^{+}_{i_{m-n}}$    annihilate $|\Lambda\rangle$. Thus, we have
\begin{eqnarray*}
   \langle\upsilon^{m}|\upsilon^{n}\rangle  &=& \langle\Lambda|E^{+}_{i_1} \cdots E^{+}_{i_{m-1}}(\sum_{k=1}^n E^{-}_{j_n}\cdots E^{-}_{j_{k+1}}h_{i_m,j_k}E^{-}_{j_{k-1}} \cdots E^{-}_{j_1})|\Lambda\rangle\\
&=& \langle\Lambda|E^{+}_{i_1} \cdots E^{+}_{i_{m-n}}(\sum h_{\mathbf{i}, \mathbf{j}}\delta_{\mathbf{i}, \mathbf{j}})|\Lambda\rangle.
\end{eqnarray*}
We draw the conclusion.
\end{proof}

Before examining the inner product of different states at the same level, we present the following lemma.
\begin{lemma}\label{quale}
If the states are obtained  by   the same contents of  operators $E^{-}_{*}$ acting on the highest-weight with different orders, the  weights corresponding to the states  are the same.
\end{lemma}
\begin{proof}
Let two states be generated by the operators $E^{-}_{1}, \cdots, E^{-}_{n}$ acting on the highest weight $\Lambda$ in different sequences. The weights of these two states are the same due to the equality
 $$w^1=w^2=\Lambda-\alpha_1-\cdots-\alpha_n.$$
\end{proof}

We then state the following theorem.
\begin{theorem}\label{ii2}
The inner product  of different states with the same level is zero.
\end{theorem}
\begin{proof}
 $|\upsilon^{n,1}_{w^1}\rangle$ and $|\upsilon^{n,2}_{w^2}\rangle$ are two different states with the same level $n$ and $w^1 \neq w^2$.
According to Lemma \ref{quale},  they  contain the same number of  operators $E^{-}_{*}$ but  have different contents.   The inner product is,
\begin{eqnarray*}
   && \langle\Lambda|E^{+}_{i_1} \cdots E^{+}_{i_n}|E^{-}_{j_n}\cdots E^{-}_{j_1}|\Lambda\rangle \\
   &=& \langle\Lambda|E^{+}_{i_1} \cdots E^{+}_{i_{n-1}} (H_{{i_n,j_n}}\delta_{i_n,j_n}+ E^{-}_{j_n}E^{+}_{i_n})E^{-}_{j_{n-1}}\cdots E^{-}_{j_1}|\Lambda\rangle \\
   &=& \langle\Lambda|E^{+}_{i_1} \cdots E^{+}_{i_{n-1}}  (\sum_{k=1}^n E^{-}_{j_n}\cdots E^{-}_{j_{k+1}}H_{i_k,j_k}\delta_{i_k,j_k}E^{-}_{j_{k-1}} \cdots E^{-}_{j_1})|\Lambda\rangle \\
      &=& \langle\Lambda|E^{+}_{i_1} \cdots E^{+}_{i_{n-1}}  (\sum_{k=1}^n E^{-}_{j_n}\cdots E^{-}_{j_{k+1}}h_{i_k,j_k}\delta_{i_k,j_k}E^{-}_{j_{k-1}} \cdots E^{-}_{j_1})|\Lambda\rangle \\
   &=& \sum h_{\mathbf{i}, \mathbf{j}}\delta_{\mathbf{i}, \mathbf{j}}\\
   &=& 0.
\end{eqnarray*}
The last equality is zero because each term of the fourth equality vanish.
\end{proof}

\subsection{Inner Product of the Same State}

In this section, we examine the inner product of states defined by different paths, which terminate with the same weight in the weight diagram.
Considering a special case of identity (\ref{le0}),
\begin{eqnarray}\label{le}
  E_{i}^{+}(E_{i}^{-})^n |\Lambda\rangle
    &=&(H_i+ E_{i}^{-}E_{i}^{+})(E_{i}^{-})^{n-1} |\Lambda\rangle\nonumber\\
  &=& \sum_{i=1}^{n-1} (E_{i}^{-})^{l} H_i (E_{i}^{-})^{n-1-l} |\Lambda\rangle\nonumber \\
    &=&\sum_{i=1}^{n-1}(E_{i}^{-})^{l}(\lambda_i-(n-1-l)A_{ii})(E_{i}^{-})^{n-1-l}|\Lambda\rangle \nonumber\\
       &=&n(\lambda_i-(n-1))(E_{i}^{-})^{n-1}|\Lambda\rangle\nonumber \\
       &=&R.H.S
\end{eqnarray}
where we have used formula (\ref{com}) in the third step.

Using the above formula, we can further  generalize  identity  (\ref{le0}).
\begin{prop}
For the highest-weight $\Lambda=\Sigma_a \lambda_a \omega_a$, we have
\begin{equation*}\label{pro}
  E_{i}^{+}(E_{i}^{-})^n \prod_{b=1}^{m}E_{j_b}^{-} |\Lambda\rangle
  =n(\lambda_i-(n-1)-\sum_{b=1}^{m}A_{j_b,i}) (E_{i}^{-})^{n-1} \prod_{b=1}^{m}E_{j_b}^{-}|\Lambda\rangle+ (E_{i}^{-})^{n} E_{i}^{+}\prod_{b=1}^{m}E_{j_b}^{-}|\Lambda\rangle.
\end{equation*}
\end{prop}
\begin{proof}: According to   Eq.(\ref{le}), we have
\begin{eqnarray*}
  L.H.S&=& \sum_{a=0}^{n-1} (E_{i}^{-})^{a} H_i (E_{i}^{-})^{n-1-a} \prod_{b=1}^{m}E_{j_b}^{-}|\Lambda\rangle+ (E_{i}^{-})^{n} E_{i}^{+}\prod_{b=1}^{m}E_{j_b}^{-}|\Lambda\rangle\\
    &=&\sum_{a=0}^{n-1} (E_{i}^{-})^{a} (\lambda_i-(n-1-a)A_{ii}-\sum_{b=1}^{m}A_{j_b,i}) (E_{i}^{-})^{n-1-a} \prod_{b=1}^{m}E_{j_b}^{-}|\Lambda\rangle\\
    &&\quad \quad \quad\quad + (E_{i}^{-})^{n} E_{i}^{+}\prod_{b=1}^{m}E_{j_b}^{-}|\Lambda\rangle \\
       &=&n(\lambda_i-(n-1)-\sum_{b=1}^{m}A_{j_b,i}) (E_{i}^{-})^{n-1} \prod_{b=1}^{m}E_{j_b}^{-}|\Lambda\rangle+ (E_{i}^{-})^{n} E_{i}^{+}\prod_{b=1}^{m}E_{j_b}^{-}|\Lambda\rangle.
\end{eqnarray*}
\end{proof}
For $n=0$, this formula reduce to Eq.(\ref{le0}). For $m=0$, we recover Eq.(\ref{le}). Using this formula, we can calculate any inner product directly; however, the computational efficiency is not acceptable.
With the notation $ |\Lambda^1\rangle=\prod_{b=1}^{m}E_{j_b}^{-} |\Lambda\rangle$,  the identity in Proposition \ref{pro} becomes
\begin{equation}\label{bfi}
  E_{i}^{+}(E_{i}^{-})^n \prod_{b=1}^{m}E_{j_b}^{-} |\Lambda\rangle
  =n(\Lambda^1-(n-1)) (E_{i}^{-})^{n-1} |\Lambda^1\rangle+ (E_{i}^{-})^{n} E_{i}^{+}\mid_{|\Lambda^1\rangle}\prod_{b=1}^{m}E_{j_b}^{-}|\Lambda\rangle.
\end{equation}
The subscript ${}_{|\Lambda^1\rangle}$ denotes the  state defined by  the operators on the right  acting  on  $|\Lambda\rangle$
$$|\Lambda^1\rangle=\prod_{b=1}^{m}E_{j_b}^{-}|\Lambda\rangle.$$

The following formula is the basis of our main result,  which is a   generalization of  the identities  (\ref{bfi}).
\begin{prop}
Let
\begin{equation}\label{lambb}
|\lambda\rangle=\prod_{b=1}^{m_1}E_{j_b}^{-} (E_{i}^{-})^{n_1}{}_{|\Lambda^1\rangle} \prod_{b=1}^{m_2}E_{j_b}^{-}(E_{i}^{-})^{n_2} {}_{|\Lambda^2\rangle} \cdots \prod_{b=1}^{m_{l}}E_{j_b}^{-}(E_{i}^{-})^{n_l} {}_{|\Lambda^l\rangle} \prod_{b=1}^{m_{l+1}}E_{j_b}^{-}|\Lambda\rangle
\end{equation}
 with  the highest-weight is given by  $\Lambda=\Sigma_a \Lambda_a \omega_a$. And $(E_{i}^{-})^{n_1},\cdots, (E_{i}^{-})^{n_l}$ are  the product factors related to the operator $E_{i}^{-}$.
Then  we have
\begin{eqnarray}\label{prop}
&&E_{i}^{+}\prod_{b=1}^{m_1}E_{j_b}^{-} (E_{i}^{-})^{n_1}{}_{|\Lambda^1\rangle} \prod_{b=1}^{m_2}E_{j_b}^{-}(E_{i}^{-})^{n_2} {}_{|\Lambda^2\rangle} \cdots \prod_{b=1}^{m_{l}}E_{j_b}^{-}(E_{i}^{-})^{n_l} {}_{|\Lambda^l\rangle} \prod_{b=1}^{m_{l+1}}E_{j_b}^{-}|\Lambda\rangle  \\
  =&&\sum_{k=1}^l\prod_{b=1}^{m_1}E_{j_b}^{-} (E_{i}^{-})^{n_1} \cdots n_k(\Lambda^k_i-(n_k-1))\prod_{b=1}^{m_k}E_{j_b}^{-}(E_{i}^{-})^{n_k-1}  \cdots \prod_{b=1}^{m_{l}}E_{j_b}^{-}(E_{i}^{-})^{n_l} \prod_{b=1}^{m_{l+1}}E_{j_b}^{-}|\Lambda\rangle,\nonumber
\end{eqnarray}
where $i\neq j_b$.
\end{prop}

\begin{proof}
According to Eq.(\ref{cn}), we have $[E_{i}^{+}, E_{j_b}^{-}]=0$ for $i\neq j_b$. According to   Eq.(\ref{le}), we have
\begin{eqnarray*}
  L.H.S&=& E_{i}^{+}\prod_{b=1}^{m_1}E_{j_b}^{-} (E_{i}^{-})^{n_1}{|\Lambda^1\rangle}\\
       &=& n_1(\Lambda^1_i-(n_1-1))\prod_{b=1}^{m_1}E_{j_b}^{-} (E_{i}^{-})^{n_1-1}{|\Lambda^1\rangle}+\prod_{b=1}^{m_1}E_{j_b}^{-} (E_{i}^{-})^{n_1}E_{i}^{+}{|\Lambda^1\rangle}\nonumber\\
        &=& n_1(\Lambda^1_i-(n_1-1))\prod_{b=1}^{m_1}E_{j_b}^{-} (E_{i}^{-})^{n_1-1}{|\Lambda^1\rangle}+\cdots\nonumber\\
        &&+\prod_{b=1}^{m_1}E_{j_b}^{-} (E_{i}^{-})^{n_1} \cdots n_k(\Lambda^k_i-(n_k-1))\prod_{b=1}^{m_k}E_{j_b}^{-}(E_{i}^{-})^{n_k-1} {|\Lambda^k\rangle}\\
        &&+\prod_{b=1}^{m_1}E_{j_b}^{-} (E_{i}^{-})^{n_1} \cdots \prod_{b=1}^{m_k}E_{j_b}^{-}(E_{i}^{-})^{n_k}E_{i}^{+} {|\Lambda^k\rangle}\nonumber\\
       &=& \cdots\\
        &=& R.H.S.
\end{eqnarray*}
\end{proof}

Fortunately, there are not always numerous terms  on the right hand side of the identity (\ref{prop}).
Let $|\Lambda^t\rangle$ be the first state satisfying   $E_{i}^{+} {|\Lambda^t\rangle}=0$  in  $\lambda$ (\ref{lambb}). Consequently,  the sequences will terminate at the $t$th factor containing $E_{i}^{-}$.  And then we have
\begin{eqnarray}\label{propp}
&&E_{i}^{+}\prod_{b=1}^{m_1}E_{j_b}^{-} (E_{i}^{-})^{n_1}{}_{|\Lambda^1\rangle} \prod_{b=1}^{m_2}E_{j_b}^{-}(E_{i}^{-})^{n_2} {}_{|\Lambda^2\rangle} \cdots \prod_{b=1}^{m_{l}}E_{j_b}^{-}(E_{i}^{-})^{n_l} {}_{|\Lambda^l\rangle} \prod_{b=1}^{m_{l+1}}E_{j_b}^{-}|\Lambda\rangle  \\
  =&&\sum_{k=1}^t\prod_{b=1}^{m_1}E_{j_b}^{-} (E_{i}^{-})^{n_1} \cdots n_k(\Lambda^k_i-(n_k-1))\prod_{b=1}^{m_k}E_{j_b}^{-}(E_{i}^{-})^{n_k-1}  \cdots \prod_{b=1}^{m_{t}}E_{j_b}^{-}(E_{i}^{-})^{n_t}|\Lambda^t\rangle.\nonumber
\end{eqnarray}

\begin{figure}[!ht]
  \begin{center}
    \includegraphics[width=6.8in]{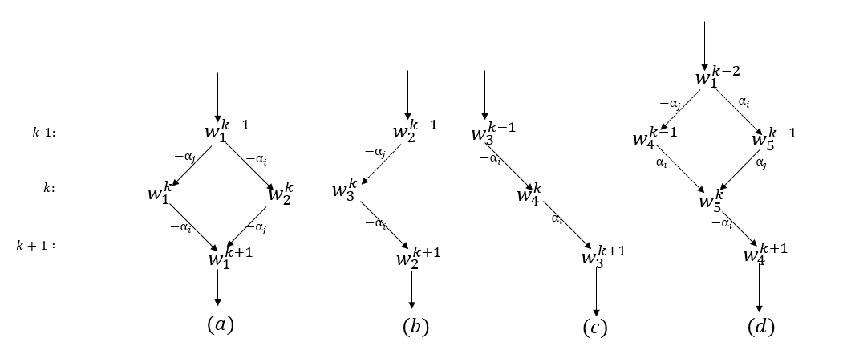}
  \end{center}
  \caption{ Weights of different levels  formed by two positive roots $\alpha_i$ and $\alpha_j$. }
  \label{t2}
\end{figure}
Moreover, we find that the terms on the right hand side of the identity (\ref{propp})   correspond  to  paths in the   weight diagram.
\begin{theorem}\label{ite}
The  action of the  $E^{+}_*$ would reduce a state corresponding  to a path in the   weight diagram  to states  corresponding  to paths in the  same  weight diagram.
\end{theorem}
\begin{proof}
We prove the theorem by induction on the level of the weight.  Assuming  the proposition is true for  weights of  level $k$, we prove the theorem for  weights of level $k+1$. The weight at the level $k+1$ is denoted as $w^{k+1}_{*}$.
For a weight of level $k+1$, the  case where only one arrow pointing it are shown in Fig.(\ref{t2})(b),(c),and (d).   The case where  at least  two arrows pointing it is shown in Fig.(\ref{t2})(a).

For the first case with one arrow, we have
$$|w^{k+1}\rangle=E^{-}_i|w^{k}\rangle,$$
where $|w^{k}\rangle$ is a state corresponding to a  path from the highest-weight state to the weight of level $k$. Then the inner product is
$$\langle w^{k+1}|w^{k+1}\rangle=\langle w^{k}|E^{+}_i|w^{k+1}\rangle.$$
We need to prove that the expansion of $E^{+}_i|w^{k+1}\rangle$ are states  corresponding to paths  from the highest-weight state to the weight of level $k$.
\begin{itemize}
  \item For  $E^{+}_i|w^{k}_3\rangle=0$ as shown in Fig.(\ref{t2})(b), according to the assumption made,   we have
 $$E^{+}_i|w^{k+1}_2\rangle=E^{+}_iE^{-}_i|w^{k}_3\rangle=H_i|w^{k}_3\rangle=h_i|w^{k}_3\rangle,$$
which leads to  the conclusion.
  \item For $E^{+}_i|w^{k}_4\rangle \neq 0$ and  $|w^{k}_4\rangle= E^{-}_i|w^{k-1}_3\rangle $ as shown in Fig.(\ref{t2})(c), according to the assumption made,  we have
\begin{equation}\label{oo1}
 E^{+}_i|w^{k}_4\rangle=a_1|w^{{k-1,j_1}}\rangle +a_2|w^{{k-1,j_2}}\rangle+\cdots+a_n|w^{{k-1,j_n}}\rangle.
\end{equation}
The states on the right hand side are states corresponding to paths from the highest-weight state to the weights of level $k-1$. Then
\begin{eqnarray*}
 E^{+}_i|w^{k+1}_3\rangle
&=&E^{+}_iE^{-}_i|w^{k}_4\rangle\\
 &=&(H_i+E^{-}_iE^{+}_i)|w^{k}_4\rangle\\
 &=&h_i|w^{k}_4\rangle+a_1E^{-}_i|w^{{k-1,j_1}}\rangle +a_2E^{-}_i|w^{{k-1,j_2}}\rangle+\cdots+a_nE^{-}_i|w^{{k-1,j_n}}\rangle,
\end{eqnarray*}
where we have used the expansion of $E^{+}_i|w^{k}_4\rangle$ (\ref{oo1}). The states on the last equality  are states corresponding to paths from the highest-weight state to the weights of level $k$. Thus we draw the conclusion.
   \item For $E^{+}_i|w^{k}_4\rangle \neq 0$ and  $|w^{k}_5\rangle= E^{-}_j|w^{k-1}_5\rangle, \, (i\neq j ) $,  then $|w^{k}_5\rangle= E^{-}_j E^{-}_i|w^{k-2}_1\rangle  $ as shown in Fig.(\ref{t2})(d).
According to the assumption made,  we have
\begin{eqnarray*}
  &&\langle w^{k,1}_5|w^{k,2}_5\rangle\nonumber \\
  &=& \langle w^{k-1,1}_4| E^{+}_i E^{-}_j |w^{k-1,2}_5\rangle \\
   &=& \langle w^{k-2,1}_1| E^{+}_j E^{+}_i E^{-}_j E^{-}_i|w^{k-2,2}_1\rangle\nonumber \\
     &=& \langle w^{k-2,1}_1| E^{+}_jE^{-}_j E^{+}_i  E^{-}_i|w^{k-2,2}_1\rangle.\nonumber
\end{eqnarray*}
 According to the assumption made,  we have
 $$E^{+}_i  E^{-}_i|w^{k-2,2}_1\rangle
   =  a'_1|w^{{k-2,j_1}}\rangle +a'_2|w^{{k-2,j_2}}\rangle+\cdots+a'_n|w^{{k-2,j_n}}\rangle.$$
 The states on the right hand side correspond to paths from the highest-weight state to the weights of level $k-2$.
As shown in Fig.(\ref{t2})(d), we have
\begin{eqnarray}\label{oo3}
  &&\langle w^{k,1}_5|w^{k,2}_5\rangle\nonumber \\
   &=& \langle w^{k-2,1}_1| E^{+}_j E^{-}_j (a'_1|w^{{k-2,j_1}}\rangle +a'_2|w^{{k-2,j_2}}\rangle+\cdots+a_n|w^{{k-2,j_n}}\rangle) \\
   &=& \langle  w^{k-2,1}_1|E^{+}_j( a_1|w^{{k-1,j_1}}\rangle +a_2|w^{{k-1,j_2}}\rangle+\cdots+a_n|w^{{k-1,j_n}}\rangle)\nonumber
\end{eqnarray}
where $|w^{{k-1,j_i}}\rangle$  on the right hand side are states corresponding to paths from the highest-weight state to the weights of level $k-1$. Then
\begin{eqnarray*}
 E^{+}_i|w^{k+1}_4\rangle&=&E^{+}_iE^{-}_iE^{-}_j E^{-}_i|w^{k-2}_1\rangle=(H_i+E^{-}_iE^{+}_i)E^{-}_j E^{-}_i|w^{k-2}_1\rangle\\
 &=&h_iE^{-}_j E^{-}_i|w^{k-2}_1\rangle+a_1E^{-}_i|w^{{k-1,j_1}}\rangle +a_2E^{-}_i|w^{{k-1,j_2}}\rangle+\cdots+a_nE^{-}_i|w^{{k-1,j_n}}\rangle,
\end{eqnarray*}
where we have used the expansion of $E^{+}_i|w^{k}_5\rangle$ (\ref{oo3}). The states on the last equality are states corresponding to paths from the highest-weight state to the weights of level $k$.  We draw the conclusion.
\end{itemize}

%%%%%%%%%%%%%%%%%%%%%%%%%%%%%%%%%%%%%%%%%%%%%%%%%%%%%%%%%%%%%%%%%%%%%
Next we consider case where   two arrows  pointing  to the state $w^{{k+1}}$ as shown in Fig.(\ref{t2})(a).
$|w^{k,1}\rangle$ and  $|w^{k,2}\rangle$ are states of level $k$ satisfying
$$|w^{k}_1\rangle=E^{-}_j|w^{{k-1}}_1\rangle,\quad\quad|w^{k}_2\rangle=E^{-}_i|w^{{k-1}}_1\rangle.$$
According to the  assumption made, we have
\begin{eqnarray}\label{expandd1}
&&E^{+}_j|w^{k}_1\rangle=E^{+}_jE^{-}_j|w^{{k-1,1}}_1\rangle=a_1|w^{{k-1,j_1}}\rangle +a_2|w^{{k-1,j_2}}\rangle+\cdots+a_n|w^{{k-1,j_n}}\rangle,
\end{eqnarray}
\begin{eqnarray}\label{expandd2}
&& E^{+}_i|w^{k}_2\rangle=E^{+}_iE^{-}_i|w^{{k-1,2}}_1\rangle=b_1|w^{{k-1,i_1}}\rangle +b_2|w^{{k-1,i_2}}\rangle+\cdots+b_m|w^{{k-1,i_m}}\rangle.
\end{eqnarray}
For the inner product $\langle w^{k+1,2}|w^{k+1,1}\rangle =  \langle w^{k,2}_{1}| E^{+}_{i}E^{-}_{j}|w^{k,1}_{2}\rangle$,  we have  to prove $ E^{+}_{i}E^{-}_{j}|w^{k,1}_{2}\rangle$ are states  corresponding to paths  from the highest-weight state to the weight of level $k$.
We have
\begin{eqnarray*}
  \langle w^{k+1,2}|w^{k+1,1}\rangle &=&  \langle w^{k,2}_{1}| E^{+}_{i}E^{-}_{j}|w^{k,1}_{2}\rangle\\
&=& \langle w^{k-1,2}_{1}|E^{+}_{j} E^{+}_{i}E^{-}_{j}E^{-}_{i}|w^{k,1}_{2}\rangle\\
&=&\langle w^{k-1,2}_{1}|E^{+}_{j} E^{-}_{j}E^{+}_{i}E^{-}_{i}|w^{k,1}_{2}\rangle  \\
 &=& \langle w^{k-1,2}_{1}|E^{+}_{j} E^{-}_{j}(b_1|w_{{k-1}_{i_1}}\rangle +b_2|w^{{k-1,i_2}}\rangle+\cdots+b_m|w^{{k-1,i_m}}\rangle)\\
  &=&\langle w^{k-1,2}_{1}|E^{+}_{j}(b_1 E^{-}_{j}|w^{{k-1,i_1}}\rangle +b_2 E^{-}_{j}|w^{{k-1,i_2}}\rangle+\cdots+b_m E^{-}_{j}|w^{{k-1,i_m}}\rangle),
\end{eqnarray*}
where we have used the expansion (\ref{expandd2}).  The states  of the last equality in brackets correspond to paths from the highest-weight state to the weights of level $k$. Thus we draw the conclusion.

For the inner product $\langle w^{k+1,1}|w^{k+1,2}\rangle =  \langle w^{k,1}_{2}| E^{+}_{j}E^{-}_{i}|w^{k,2}_{1}\rangle$,  we have  to prove $ E^{+}_{j}E^{-}_{i}|w^{k,2}_{1}\rangle$ are states  corresponding to paths  from the highest-weight state to the weight of level $k$.
Then we get
\begin{eqnarray*}
 \langle w^{k+1,1}|w^{k+1,2}\rangle &=&  \langle w^{k,1}_{1}| E^{+}_{i}E^{-}_{j}|w^{k,2}_{1}\rangle\\
&=& \langle w^{k-1,1}_1|E^{+}_{i}E^{+}_{j} E^{-}_{i}E^{-}_{j}|w^{k-1,2}_1\rangle\\
&=&\langle w^{k-1,1}_1|E^{+}_{i} E^{-}_{i}E^{+}_{j}E^{-}_{j}|w^{k-1,2}_1\rangle  \\
 &=& \langle w^{k-1,1}_1|E^{+}_{i} E^{-}_{i}(a_1|w^{{k-1,j_1}}\rangle +a_2|w^{{k-1,j_2}}\rangle+\cdots+a_n|w^{{k-1,j_n}}\rangle)\\
  &=& \langle w^{k-1,1}_1|E^{+}_{i}(a_1E^{-}_{i}|w^{{k-1,j_1}}\rangle +a_2E^{-}_{i}|w^{{k-1,j_2}}\rangle+\cdots+a_nE^{-}_{i}|w^{{k-1,j_n}}\rangle),
\end{eqnarray*}
where we have used the expansion (\ref{expandd1}). The states  of the last equality in brackets correspond to path from the highest-weight state to the weights of level $k$. Thus we draw the conclusion.

For the inner product of $\langle w^{k+1,1}|w^{k+1,1}\rangle$ and  $\langle w^{k+1,2}|w^{k+1,2}\rangle$, the proofs  are reduced to the first case.
\end{proof}
\begin{rema}
  The actions of the operators reveal the global  information of the weights of the  highest-weight representations.
\end{rema}

By using  formula (\ref{propp}) again and again, the algorithm for calculating the inner product present itself now.
\begin{algorithm}\label{algo}
The states $ |\upsilon^1_{w}\rangle$ and $|\upsilon^2_{w}\rangle$  are  defined by  paths from the highest-weight vector $\Lambda$ to the weight $w$. To calculate  the inner product,
 \begin{equation}
    \langle \upsilon^1_{w}|\upsilon^2_{w}\rangle=\langle\Lambda_s|E^{+}_{i_1} \cdots E^{+}_{i_n}|E^{-}_{j_n}\cdots E^{-}_{j_1}|\Lambda_s\rangle,\nonumber
\end{equation}
 the operators defining  the state  $\langle \upsilon^1_{w}|$ act on the state $|\upsilon^2_{w}\rangle$ in sequence from right to left  using formula (\ref{prop}).  The  actions of the raising operators would reduce the states on the right  to the states  one to one correspondence  with  paths  in the weight diagram. These processes continue until no operators are left.
\end{algorithm}
\begin{flushleft}
The algorithm is an iterative process because of Theorem \ref{ite}. Thus the  efficiency of calculations of inner product  is improved greatly.
\end{flushleft}

Finally, we would like to point out a fact.
\begin{theorem}\label{iii2}
The inner product of states depends on paths  defining the states.
\end{theorem}
\begin{figure}[!ht]
  \begin{center}
    \includegraphics[width=6.8in]{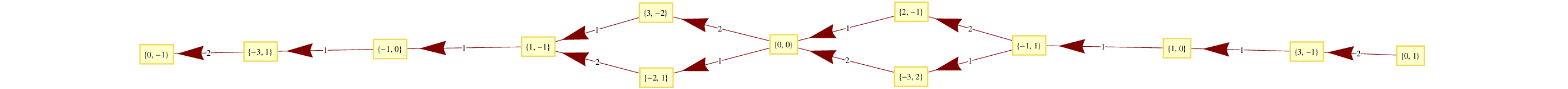}
  \end{center}
  \caption{Weights in the fundamental representation $(0,1)$ of $G_2$. }
  \label{G2}
\end{figure}
We  illustrate this fact through examples.  As shown in Fig.(\ref{G2}), there are two paths from the highest-weight state $(0,1)$ to the weight $(0,0)$. These two paths correspond to the following  states
\begin{eqnarray*}
   && |\upsilon^{5,1}_{(0,0)}\rangle=E^{-}_{2}E^{-}_{1}E^{-}_{1}E^{-}_{1}E^{-}_{2}|\Lambda_2\rangle,  \\
   && |\upsilon^{5,2}_{(0,0)}\rangle= E^{-}_{1}E^{-}_{2}E^{-}_{1}E^{-}_{1}E^{-}_{2}|\Lambda_2\rangle.
\end{eqnarray*}
The  inner products can be calculated  by commutating   the operators directly as follows
\begin{eqnarray*}
  \langle\upsilon^{5,1}_{(0,0)}|\upsilon^{5,1}_{(0,0)}\rangle &=&\langle \Lambda_2| E^{+}_{2}E^{+}_{1}E^{+}_{1}E^{+}_{1}E^{+}_{2}|E^{-}_{2}E^{-}_{1}E^{-}_{1}E^{-}_{1}E^{-}_{2}|\Lambda_2\rangle=72, \\
 \langle\upsilon^{5,2}_{(0,0)} |\upsilon^{5,1}_{(0,0)}\rangle &=&\langle \Lambda_2| E^{+}_{2}E^{+}_{1}E^{+}_{1}E^{+}_{2}E^{+}_{1}| E^{-}_{2}E^{-}_{1}E^{-}_{1}E^{-}_{1}E^{-}_{2}|\Lambda_2\rangle=24, \label{a2a2a2}\\
 \langle\upsilon^{5,2}_{(0,0)} |\upsilon^{5,2}_{(0,0)}\rangle &=&\langle \Lambda_2|E^{+}_{2}E^{+}_{1}E^{+}_{1}E^{+}_{2}E^{+}_{1}|E^{-}_{1}E^{-}_{2}E^{-}_{1}E^{-}_{1}E^{-}_{2}|\Lambda_2\rangle=36,
\end{eqnarray*}
which are different from each other.
In the appendix, we will repeatedly  calculate these inner products by using  Algorithm \ref{algo}.

\section{Applications }
In previous sections, we proposed an iterative algorithm for calculating  the inner product of states of   highest-weight  representation.   In this section, we discuss the  applications of the algorithm, such as the norm of states and the unitarity of representations.  And we  completely  determine the inner products of states of the minuscule representations.   Algorithm \ref{algo} also work  for the highest-weight representation of affine Lie algebra, which is infinite dimension. Finally, we note that it can be used to study the solution of the \textbf{KW} equations, where the inner products are related to the fundamental representations only.
\subsection{Unitarity and Norm}
As an application of Algorithm \ref{algo}, we discuss the calculations of the  norms of states. For the state $|\lambda\rangle= E_{j_{n(w)}}^-\cdots E_{j_1}^-|\Lambda_s\rangle$, the norm is
$$\langle\lambda|\lambda\rangle = \langle\Lambda_s|E_{j_1}^+\cdots  E_{j_{n(w)}}^+|E_{j_{n(w)}}^-\cdots E_{j_1}^-|\Lambda_s\rangle.$$
The  state  $|\lambda\rangle$ in the height weight representation is given by
\begin{equation}\label{sta}
|\lambda\rangle= (E_{i}^{-})^{n_1}{}_{|\Lambda^1\rangle} \prod_{b=1}^{m_2}E_{j_b}^{-}(E_{i}^{-})^{n_2} {}_{|\Lambda^2\rangle} \cdots \prod_{b=1}^{m_{l}}E_{j_b}^{-}(E_{i}^{-})^{n_l} {}_{|\Lambda^l\rangle} \prod_{b=1}^{m_{l+1}}E_{j_b}^{-}|\Lambda\rangle.
\end{equation}
The notations have the same means with that in  equality (\ref{propp}), with  the  subscript $i$ of $m_i$  begin from two.

To calculate the norm of  $|\lambda\rangle$, we need to take the  action $E_{i}^{+}|\lambda\rangle$ firstly.
Thus, we rewrite the identity (\ref{propp}) as follows
\begin{eqnarray}
&&E_{i}^{+} (E_{i}^{-})^{n_1}{}_{|\Lambda^1\rangle} \prod_{b=1}^{m_2}E_{j_b}^{-}(E_{i}^{-})^{n_2} {}_{|\Lambda^2\rangle} \cdots \prod_{b=1}^{m_{l}}E_{j_b}^{-}(E_{i}^{-})^{n_l} {}_{|\Lambda^l\rangle} \prod_{b=1}^{m_{l+1}}E_{j_b}^{-}|\Lambda\rangle  \\
  =&&\sum_{k=1}^t (E_{i}^{-})^{n_1} \cdots n_k(\Lambda^k_i-(n_k-1))\prod_{b=1}^{m_k}E_{j_b}^{-}(E_{i}^{-})^{n_k-1}  \cdots \prod_{b=1}^{m_{t}}E_{j_b}^{-}(E_{i}^{-})^{n_t}|\Lambda^t\rangle.\nonumber
\end{eqnarray}

We have the  following conjecture for the coefficients   on the right hand side of  the identity.
\begin{conj}\label{co}
$n_k(\Lambda^k_i-(n_k-1))\geq 0$.
\end{conj}
 Unfortunately, we are not able to prove it presently.  We have verified it through numerous examples  and  some of which are  given   in the appendix. The validity of the conjecture is based on the paths in the weight diagram of the highest-weight representation  defining the state $|\lambda\rangle$.
This conjecture guarantees that coefficients are positive for each step of  \textbf{Algorithm}  \ref{algo}.   Thus,  the norm is positive for any state in the representation, which implies  the unitary of the space. This also holds for linear combinations  of such states.

To calculate the norm of  $|\lambda\rangle$, we need to take the following action.
\begin{eqnarray}
&&E_{i}^{+})^{n_1}|\lambda\rangle \nonumber\\
 &=&(E_{i}^{+})^{n_1-1}E_{i}^{+} (E_{i}^{-})^{n_1}{}_{|\Lambda^1\rangle} \prod_{b=1}^{m_2}E_{j_b}^{-}(E_{i}^{-})^{n_2} {}_{|\Lambda^2\rangle} \cdots \prod_{b=1}^{m_{l}}E_{j_b}^{-}(E_{i}^{-})^{n_l} {}_{|\Lambda^l\rangle} \prod_{b=1}^{m_{l+1}}E_{j_b}^{-}|\Lambda\rangle  \nonumber\\
  &=&(E_{i}^{+})^{n_1-1} \sum_{k=l}^t (E_{i}^{-})^{n_1}{}_{|\Lambda^1+\alpha^k\rangle} \prod_{b=1}^{m_2}E_{j_b}^{-}(E_{i}^{-})^{n_2} {}_{|\Lambda^2+\alpha^k\rangle} \cdots \\
  && n_k(\Lambda^k_i-(n_k-1))\prod_{b=1}^{m_k}E_{j_b}^{-}(E_{i}^{-})^{n_k-1} {}_{|\Lambda^k\rangle}\cdots \prod_{b=1}^{m_{t}}E_{j_b}^{-}(E_{i}^{-})^{n_t}|\Lambda^t\rangle.\nonumber
\end{eqnarray}
The operator $E_{i}^{+}$  decrease one  $E_{i}^{-}$ in each product factor of  $|\lambda\rangle$  containing   $E_{i}^{-}$. After the action of operators $(E_{i}^{+})^{n_1}$,   the number of  terms on the right hand side of  the last equality is
\begin{equation}\label{comb}
  C^{k_1}_{n_1}+C^{k_2}_{n_1-k_1}+\cdots +C^{k_{t-1}}_{n_1-k_1-\cdots -k_{t-2}}+C^{k_{t}}_{n_1-k_1-\cdots -k_{t-1}}
\end{equation}
with $k_1+k_2+ \cdots +k_t=n_1$. For the last term, we have
$$C^{k_{t}}_{n_1-k_1-\cdots -k_{t-1}}=C^{k_{t}}_{k_{t}}=1.$$

There are many difficulties to get a closed formula of norm of $\lambda$. According to formula (\ref{le0}),  the operators  in the conjunction state  $\langle\lambda|$  would  decrease the same   operators    in the state $|\lambda\rangle$.
\begin{enumerate}
  \item The combination formula (\ref{comb})  imply there is no simple closed formula to describe the norm of $|\lambda\rangle$.
  \item The two adjacent product  factors of $E_{i}^{-}$  become one when the product factor  $\prod_{b=1}^{m_{l}}E_{j_b}^{-}$ are decreased by the operators in $\langle\lambda|$.  The actions of operators lead to   more complicated  combination formula   than formula  (\ref{comb}).
  \item Note that  $|\Lambda^t\rangle$ would not be annihilated  by  operators except $E_{i}^{+}$.
\end{enumerate}

\begin{flushleft}
\textbf{Special states}:  There are states corresponding to a particular type of paths in the weight diagram,  whose norms   can be fully determined.

These  states and the norms are given by the following proposition.
\end{flushleft}
\begin{prop}
For the highest-weight $\Lambda=\Sigma_a \lambda_a \omega_a$, we have
\begin{eqnarray}\label{proppp}
&&|\lambda\rangle=\prod_{b=1}^{m_1}E_{j_b}^{-} (E_{i_1}^{-})^{n_1}{}_{|\Lambda^1\rangle} \prod_{b=1}^{m_2}E_{j_b}^{-}(E_{i_2}^{-})^{n_2} {}_{|\Lambda^2\rangle} \cdots \prod_{b=1}^{m_{l}}E_{j_b}^{-}(E_{i_l}^{-})^{n_l} {}_{|\Lambda^l\rangle} \prod_{b=1}^{m_{l+1}}E_{j_b}^{-}|\Lambda\rangle,\\
&&|\lambda^{'}\rangle= \prod_{b=m_0}^{m_k}E_{j_b}^{-}(E_{i_k}^{-})^{n_k} {}_{|\Lambda^k\rangle} \cdots \prod_{b=1}^{m_{l}}E_{j_b}^{-}(E_{i_l}^{-})^{n_l} {}_{|\Lambda^l\rangle} \prod_{b=1}^{m_{l+1}}E_{j_b}^{-}|\Lambda\rangle,
\end{eqnarray}
 where the   states  satisfy the following constraints
$$E_{i_k}^{+}|\Lambda^k\rangle=0, \quad E_{j_{m_0-1}}|\lambda^{'}\rangle=0, \quad m_0=1,\cdots,m_k,\quad k=1,\cdots,l.$$
Then the  norm of $|\lambda\rangle$ is given by
$$\langle\lambda|\lambda\rangle=\prod_{k=1}^{l}(\Lambda^k_{i_k}!)^2.$$
And the  norm of $|\lambda^{'}\rangle$ is given by
$$\langle\lambda^{'}|\lambda^{'}\rangle=\prod_{b=k}^{l}(\Lambda^b_{i_b}!)^2.$$
\end{prop}
\begin{proof} According to the construction  of the path, we have $n_k=\Lambda^k_{i_k},(k=1,\cdots,l)$ and
\begin{eqnarray*}
  &&\langle\lambda|\lambda\rangle \\
&=& \langle\cdots E_{j_1}^{+}\prod_{b=1}^{m_1}E_{j_b}^{-} (E_{i_1}^{-})^{n_1}{}_{|\Lambda^1\rangle} \prod_{b=1}^{m_2}E_{j_b}^{-}(E_{i_2}^{-})^{n_2} {}_{|\Lambda^2\rangle} \cdots \prod_{b=1}^{m_{l}}E_{j_b}^{-}(E_{i_l}^{-})^{n_l} {}_{|\Lambda^l\rangle} \prod_{b=1}^{m_{l+1}}E_{j_b}^{-}|\Lambda\rangle.\\
   &=& \langle\cdots E_{j_2}^{+}\prod_{b=2}^{m_1}E_{j_b}^{-} (E_{i_1}^{-})^{n_1}{}_{|\Lambda^1\rangle} \prod_{b=1}^{m_2}E_{j_b}^{-}(E_{i_2}^{-})^{n_2} {}_{|\Lambda^2\rangle} \cdots \prod_{b=1}^{m_{l}}E_{j_b}^{-}(E_{i_l}^{-})^{n_l} {}_{|\Lambda^l\rangle} \prod_{b=1}^{m_{l+1}}E_{j_b}^{-}|\Lambda\rangle.\\
&=& \langle\cdots (E_{n_1}^{+})^{n_1-1}E_{n_1}^{+} (E_{i_1}^{-})^{n_1}{}_{|\Lambda^1\rangle} \prod_{b=1}^{m_2}E_{j_b}^{-}(E_{i_2}^{-})^{n_2} {}_{|\Lambda^2\rangle} \cdots \prod_{b=1}^{m_{l}}E_{j_b}^{-}(E_{i_l}^{-})^{n_l} {}_{|\Lambda^l\rangle} \prod_{b=1}^{m_{l+1}}E_{j_b}^{-}|\Lambda\rangle.\\
&=& \langle\cdots (E_{n_1}^{+})^{n_1-1}n_1(1) (E_{i_1}^{-})^{n_1}{}_{|\Lambda^1\rangle} \prod_{b=1}^{m_2}E_{j_b}^{-}(E_{i_2}^{-})^{n_2} {}_{|\Lambda^2\rangle} \cdots \prod_{b=1}^{m_{l}}E_{j_b}^{-}(E_{i_l}^{-})^{n_l} {}_{|\Lambda^l\rangle} \prod_{b=1}^{m_{l+1}}E_{j_b}^{-}|\Lambda\rangle.\\
&=& \langle\cdots (E_{n_2}^{+})^{n_2}\prod_{b=1}^{m_2}E_{j_b}^{+}(\Lambda^1_{i_1}!)^2 \prod_{b=1}^{m_2}E_{j_b}^{-}(E_{i_2}^{-})^{n_2} {}_{|\Lambda^2\rangle} \cdots \prod_{b=1}^{m_{l}}E_{j_b}^{-}(E_{i_l}^{-})^{n_l} {}_{|\Lambda^l\rangle} \prod_{b=1}^{m_{l+1}}E_{j_b}^{-}|\Lambda\rangle.\\
&=&\cdots\\
&=&\prod_{k=1}^{l}(\Lambda^k_{i_k}!)^2.
\end{eqnarray*}
Similarly, we can get  the  norm of $|\lambda^{'}\rangle$.
\end{proof}
\begin{example}
The states corresponding to the weights in Fig.(\ref{gcd}) belong to  states in this proposition.
\end{example}

We have a  generalization  of formula (\ref{proppp}).
The state is
$$|\lambda\rangle=(E_{i_0}^{-})^{n_0}{}_{|\Lambda^0\rangle}\prod_{b=1}^{m_1}E_{j_b}^{-} (E_{i_1}^{-})^{n_1}{}_{|\Lambda^1\rangle} \prod_{b=1}^{m_2}E_{j_b}^{-}(E_{i_2}^{-})^{n_2} {}_{|\Lambda^2\rangle} \cdots \prod_{b=1}^{m_{l}}E_{j_b}^{-}(E_{i_l}^{-})^{n_l} {}_{|\Lambda^l\rangle} \prod_{b=1}^{m_{l+1}}E_{j_b}^{-}|\Lambda\rangle,$$
 where $n_0\neq \Lambda_{i_0}^0$.
The norm of this state is
\begin{eqnarray*}
  &&\langle\lambda^{''}|\lambda^{''}\rangle \\
&=& \langle\cdots (E_{i_0}^{+})^{n_0}(E_{i_0}^{-})^{n_0}{}_{|\Lambda^0\rangle}\prod_{b=1}^{m_1}E_{j_b}^{-} (E_{i_1}^{-})^{n_1}{}_{|\Lambda^1\rangle} \prod_{b=1}^{m_2}E_{j_b}^{-}(E_{i_2}^{-})^{n_2} {}_{|\Lambda^2\rangle} \cdots \prod_{b=1}^{m_{l}}E_{j_b}^{-}(E_{i_l}^{-})^{n_l} {}_{|\Lambda^l\rangle} \prod_{b=1}^{m_{l+1}}E_{j_b}^{-}|\Lambda\rangle\\
   &=&\langle\cdots (E_{i_0}^{+})^{n_0-1}(E_{i_0}^{-})^{n_0-1}{}_{|\Lambda^0\rangle}\prod_{b=1}^{m_1}E_{j_b}^{-} (E_{i_1}^{-})^{n_1}{}_{|\Lambda^1\rangle} \prod_{b=1}^{m_2}E_{j_b}^{-}(E_{i_2}^{-})^{n_2} {}_{|\Lambda^2\rangle} \cdots \prod_{b=1}^{m_{l}}E_{j_b}^{-}(E_{i_l}^{-})^{n_l} {}_{|\Lambda^l\rangle} \prod_{b=1}^{m_{l+1}}E_{j_b}^{-}|\Lambda\rangle\\
&&\cdot(\Lambda_{i_0}^0-(n_0-1))n_0\\
&=& \prod_{k=1}^{n_0}(\Lambda_{i_0}^0-k+1)k\langle\lambda|\lambda\rangle\\
&=&\prod_{k=1}^{n_0}(\Lambda_{i_0}^0-k+1)k\prod_{k=1}^{l}(\Lambda^k_{i_k}!)^2.
\end{eqnarray*}

\subsection{Minuscule Representations}
For minuscule representations, the fundamental weight is the highest weight. Table \ref{tm} presents a complete list of minuscule fundamental weights for simple Lie algebras \cite{Green}, which are the highest-weight vectors. For the minuscule representations, all the strings in the weight spaces are two terms long,
$$(\lambda, \alpha^{\vee})=2\frac{(\lambda,\alpha)}{(\alpha,\alpha)}\leq 1.$$
In \cite{Sh06}, we conjecture a formula of the factors in the  braces of formula (\ref{kwef}) for minuscule representations. Calculating the inner product of states in these representations would help prove this conjecture.   The weights in the fundamental representation  $\rho_2$ of $A_4$ are displayed as follows.
$$
\begin{array}{ccccccc}
 (0,1,0,0) &  &  &  &  &  &  \\
  \downarrow \alpha_2  &  &  &  &  &  &  \\
 (1,-1,1,0) & \xrightarrow{\alpha_1} & (-1,0,1,0)  &  &  &  &  \\
  \downarrow \alpha_3 &  & \downarrow \alpha_3 & &  &  &  \\
  (1,0,-1,1) &\xrightarrow{\alpha_1} & (-1,1,-1,1) & \xrightarrow{\alpha_2}& (0,-1,0,1) &  &  \\
 \downarrow \alpha_4 &  & \downarrow \alpha_4  & & \downarrow \alpha_4 &  & \\
  (1,0,0,-1) &\xrightarrow{\alpha_1} & (-1,1,0,-1) & \xrightarrow{\alpha_2}& (0,-1,1,-1) & \xrightarrow{\alpha_3}& (0,0,-1,0).
\end{array}
$$

\begin{table}
\begin{tabular}{cc}\hline
Type   &   \{$i$: $\omega_i$ is minuscule \} \\
$A_l$  & $ 1,2,\cdots, l$\\
$B_l$  & $l$ \\
$C_l$  & 1  \\
$D_l$  & $1,l-1,l$  \\
 $E_6$   & 1,5 \\
 $E_7$   & 6 \\
 $E_8$   &  none\\
  $F_4$  &  none \\
 $G_2$   &  none\\  \hline
\end{tabular}
\caption{Minuscule fundamental weights for simple Lie algebras. }
\label{tm}
\end{table}

Since there are only two elements in the string of weights (\ref{string}) for minuscule representations, we have $n_k=1,(k=1,\cdots,l)$ in formula (\ref{proppp}).
Thus, the norms of the states $|\lambda^{'''}\rangle$ is
$$\langle^{'''}\lambda|\lambda^{'''}\rangle=1.$$

Furthermore,  the inner product of weights in these representations can be  completely determined. The  proof of the   following theorem show the subtleties  of  Algorithm \ref{algo} for calculating inner products.
\begin{theorem}\label{mf}
For minuscule representations, the inner product of states   defined by paths ending with  the same weight in the weight diagram   is one.
\end{theorem}
\begin{proof}
  The  states $|\upsilon^{n,i}_w\rangle$ and $|\upsilon^{n,j}_w\rangle$ corresponding to weight $w$ with level $n$   are given by
$$|\upsilon^{n,i}\rangle=|E^{-}_{i_n}\cdots E^{-}_{i_1}|\Lambda\rangle, \quad\quad |\upsilon^{n,j}\rangle=|E^{-}_{j_n}\cdots E^{-}_{j_1}|\Lambda\rangle.$$
%Since they contain the same contents of  operators $E^{-}_{*}$, we would get the same weight
%$$\upsilon^{n,j}=\upsilon^{n,i}=\Lambda-\alpha_{i_1}-\cdots-\alpha_{i_n}.$$

To prove the theorem, we need to consider two cases as shown in Fig.(\ref{t2})(a) and (b).
For the first case, the inner product of $|\upsilon^{n,i}_w\rangle$ and $|\upsilon^{n,j}_w\rangle$ is
\begin{eqnarray}\label{m1}
  \langle\upsilon^{n,j}_w |\upsilon^{n,i}_w\rangle&=& \langle\Lambda|E^{+}_{i_1} \cdots E^{+}_{i_n}|E^{-}_{j_n}\cdots E^{-}_{j_1}|\Lambda\rangle \nonumber\\
  &=& \langle\Lambda|E^{+}_{i_1} \cdots E^{+}_{i_{n-2}}(E^{+}_{j_n}E^{+}_{i_n} E^{-}_{j_n}E^{-}_{i_n})E^{-}_{j_{n-2}} \cdots E^{-}_{j_1}|\Lambda\rangle \nonumber\\
  &=& \langle\Lambda|E^{+}_{i_1} \cdots E^{+}_{i_{n-2}} E^{+}_{j_n}E^{-}_{j_n}(H_{i_n}+ E^{-}_{i_n}E^{+}_{i_n})E^{-}_{j_{n-1}}\cdots E^{-}_{j_1}|\Lambda\rangle \nonumber\\
   &=& \langle\Lambda|E^{+}_{i_1} \cdots E^{+}_{i_{n-2}} (E^{+}_{j_n}E^{-}_{j_n})E^{-}_{j_{n-2}}\cdots E^{-}_{j_1}|\Lambda\rangle \nonumber\\
 &=& \langle\Lambda|E^{+}_{i_1} \cdots E^{+}_{i_{n-2}} E^{-}_{j_{n-2}}\cdots E^{-}_{j_1}|\Lambda\rangle    \nonumber\\
 &=& 1.
\end{eqnarray}
For the second case, the inner product is
\begin{eqnarray}\label{m2}
    \langle\upsilon^{n,j}_w |\upsilon^{n,i}_w\rangle&=& \langle\Lambda|E^{+}_{i_1} \cdots E^{+}_{i_n}|E^{-}_{j_n}\cdots E^{-}_{j_1}|\Lambda\rangle \nonumber\\
  &=& \langle\Lambda|E^{+}_{i_1} \cdots E^{+}_{i_{n-1}}(E^{+}_{i_n} E^{-}_{i_n})E^{-}_{j_{n-1}} \cdots E^{-}_{j_1}|\Lambda\rangle \nonumber\\
  &=& \langle\Lambda|E^{+}_{i_1} \cdots E^{+}_{i_{n-1}} (H_{i_n}+ E^{-}_{i_n}E^{+}_{i_n})E^{-}_{j_{n-1}}\cdots E^{-}_{j_1}|\Lambda\rangle \nonumber\\
   &=& \langle\Lambda|E^{+}_{i_1} \cdots E^{+}_{i_{n-1}} E^{-}_{j_{n-1}}\cdots E^{-}_{j_1}|\Lambda\rangle \nonumber\\
 &=& 1.
\end{eqnarray}
%Repeating the above processes leading to   formulas (\ref{m1}) and  (\ref{m1}) again and again,
We draw the conclusion.
\end{proof}
In \cite{Sh06}, another proof of Theorem \ref{mf} is provided by induction on the level of weights.

\subsection{Affine Lie Algebras}
The algorithm described in \ref{algo} for calculating the inner product is applicable  for affine Lie algebra $\hat{\mathfrak{g}}$. The Dynkin diagram of $\hat{\mathfrak{g}}$ is obtained from that of $\mathfrak{g}$ by adding  an extra node representing the  extra simple root $\alpha_0$.
Given a set of affine simple roots and a scalar product \cite{cft}, the extended Cartan matrix of affine Lie algebras is defined by
$$\hat{A}_{ji}=(\alpha_i,\alpha_j^{\vee})\quad\quad 0\leq i , j\leq r.$$
The matrix $\hat{A}$ encodes the whole structure of $\hat{\mathfrak{g}}$.
In the Chevally basis, the communication relations for the generators associated with the simple roots of $\hat{\mathfrak{g}}$ can be written as
\begin{equation}\label{acn}
 [E_i^{+}, E_{j}^{-}] = \delta_{ji}H_{j}, \quad  [H_{i}, E^{\pm}_{j}] = \pm \hat{A}_{ji}E^{\pm}_{j}, \quad [H_i, H_j] =0,
\end{equation}
where now $i,j=0,1,\cdots,r$. The communication relations are the same as those in formula (\ref{cn}), with the exception of the operators corresponding to  $\alpha_0$.
However,  this formulation does not explicitly reveal the infinite-dimensional nature of $\hat{\mathfrak{g}}$.

The procedure  that lists the weights in irreducible highest-weight representations of $\mathfrak{g}$ also works for $\hat{\mathfrak{g}}$. We simply have to keep track of an additional Dynkin label. However, this algorithm does not terminate in the affine case.

Since the communication relation for the generators and the construction of states are the same as those for the semisimple Lie algebra $\mathfrak{g}$,   Algorithm \ref{algo} for calculating the inner product of states in the highest-weight representations is applicable to the affine Lie algebra $\hat{\mathfrak{g}}$.

%Note that the eigenvalue of the operator $H_i$ is positive for a state satisfying $E^{+}_i{|\Lambda^k\rangle}=0$
%\begin{equation}\label{ph}
%  H_i{|\Lambda^k\rangle}=\Lambda^k_i{|\Lambda^k\rangle}.
%\end{equation}
%We act the operator $E^{-}_{i}$ only if $\Lambda^k_i>0$ according to the rules ??????????????.

\subsection{Kapustin-Witten Equations}\label{kw}
We begin by reviewing the Kapustin-Witten equations. An extensive introduction to this topic can be found in \cite{vm}.

%arising from the  interpretation of the geometric Langlands program in terms of electric-magnetic duality of quantum gauge theory.

The  localization equations of the  twisted $N=4$ super Yang-Mills theory can be  applied to  the description of the Khovanov homology of knots \cite{ComplexCS,5knots, GaiottoWitten}. On a half space  $V=\mathbb{R}^3\times \mathbb{R}_+$,  the supersymmetry conditions lead to the  Kapustin-Witten  equations.
As described in \cite{5knots},  the \textbf{KW} equations  are
\begin{eqnarray}
F-\phi\wedge\phi+\ast{ d}_A\phi =0=d_A\ast\phi\,,\label{bogomol1}
\end{eqnarray}
where $d_A$ is the covariant exterior derivative associated with a connection $A$, and $\phi$ is  one-form valued in the adjoint of the gauge group $G$. There is a Lie product understood in the $\phi\wedge\phi$ terms and $\ast$ denotes the Hodge duality.
Different reductions of the  equations lead to other well known equations {\it e.g.}, Nahm's equations, Bogomolny equations or Hitchin equations.

The solutions of the model   were studied  in \cite{5knots}  \cite{Henningson2} \cite{Henningson}.
As we all know, all kinds of Toda systems were studied a lot recently in \cite{MPAG} and the Toda systems have very nice structures from the point of solutions. For any simple compact gauge group,  after reducing to a Toda system,
in \cite{vm}, V. Mikhaylov conjectured a formula of the solutions of the model  for the boundary 't~Hooft operator.
The 't Hooft operator corresponds to  a cocharacter $\omega\in\Gamma^\vee_{ch}$.  Let $\Delta$ be the set of simple roots $\alpha_i$, and then $\alpha_i(\hat{\omega})=m_i$ with $\hat{\omega}\equiv\omega+\delta^\vee$, $\delta^\vee\in \mathfrak{b}$ is the dual of the Weyl vector in the sense that $\alpha_i(\delta^\vee)=1$.  $E_\alpha$ are the raising generators corresponding to the simple roots, and  then the explicit fields on the solution are
\begin{eqnarray}
&&\phi_0=-\frac{i}{2\rho}\partial_\sigma\chi(\sigma)\,,\nonumber\\
&&\varphi=\frac{1}{r}\sum_{\alpha\in\Delta}\exp\left[\alpha(i\omega\theta+\frac12\chi(\sigma))\right]E_\alpha\,,\nonumber\\
&&A=-i\left(\hat{\omega}+\frac12\frac{y}{\sqrt{y^2+r^2}}\partial_\sigma\chi(\sigma)\right){\rm d}\theta\, ,\nonumber
\end{eqnarray}
where $\chi(\sigma)=\sum \chi_i(\sigma)H_i$ with coroots  $H_i$. The functions $\chi_i(\sigma)$ are conjectured as follows,
\begin{eqnarray}\label{kwef}
 &&\rm e^{-\chi_s(\sigma)}\\
 &=&2^{-B_s}\sum_{w\in\Delta_s}\left[\exp\left(2\sigma w(\hat{\omega})\right)\,\langle v_w(\hat{\omega})|v_w(\hat{\omega})\rangle\,(-1)^{n(w)}\,\prod_{\beta_a\in\Delta_+}\left(\beta_a(\hat{\omega})\right)^{-2\langle w,\beta_a\rangle/\langle\beta_a,\beta_a\rangle}\right],\nonumber
\end{eqnarray}
where $B_s=2\sum_j  A_{s,j}^{-1}$ with Cartan matrix $A_{i,j}$. For a weight $w=\Lambda_s-\sum\limits_{l=1}^{n(w)}\alpha_{j_l},\ \  \alpha_{j_l} \in \Delta$ in the fundament representation  $\rho_s$,  the  vector  $|\upsilon_{w}(\hat{\omega})\rangle$  is defined as follows
$$|\upsilon_{w}(\hat{\omega})\rangle=\sum\limits_\textrm{\textbf{s}}\prod\limits_{a=1}^{n(w)}\frac{1}{w(\hat{\omega})-w_a(\hat{\omega})}
E_{j_{n(w)}}^-\cdots E_{j_1}^-|\Lambda_s\rangle,$$
where $\textrm{\textbf{s}}$ enumerate ways in which the weight $w$ can be reached from the highest-weight, i.e each $\textrm{\textbf{s}}$ corresponds to the following sequences
$$E_{j_{n(w)}}^-\cdots E_{j_1}^-|\Lambda_s\rangle.$$
In order to prove this conjecture, we  need to check the following boundary condition
\begin{equation}\label{bbc}
 e^{-\chi_s(\sigma)}\mid_{\sigma\rightarrow 0}=0.
\end{equation}
 For other related  work on these equations, see \cite{uhlenbeck}\cite{M-Witten}\cite{He}\cite{Yuuji}\cite{Teng}.

To check the boundary condition (\ref{bbc}), we have to compute the following inner product
$$W_w= \langle \upsilon_{w}(\hat{\omega})|\upsilon_{w}(\hat{\omega})\rangle,$$
which involves the inner products
\begin{equation*}
  \langle\upsilon^{i}_{w}|\upsilon^{j}_{w}\rangle=\langle\Lambda_s|E^{+}_{i_1} \cdots E^{+}_{i_n}|E^{-}_{j_n}\cdots E^{-}_{j_1}|\Lambda_s\rangle,
\end{equation*}
where $E^{-}_{j_n}\cdots E^{-}_{j_1}$ and $E^{+}_{i_1} \cdots E^{+}_{i_n}$  are two  sequences from the highest-weight $\Lambda_s$ to weight $w$.

Algorithm \ref{algo} and the results in previous subsection would be helpful for the calculation of the inner product
\begin{equation*}
  \langle\upsilon^{i}_{w}|\upsilon^{j}_{w}\rangle=\langle\Lambda_s|E^{+}_{i_1} \cdots E^{+}_{i_n}|E^{-}_{j_n}\cdots E^{-}_{j_1}|\Lambda_s\rangle,
\end{equation*}
In \cite{Sh06}, we give another form of the factors in the square bracket of formula (\ref{kwef})  for certain weights of highest-weight representations of simply laced Lie algebras, which bypass the above  computation difficulties.

\section*{Acknowledgments}
Chuanzhong Li is supported by the National Natural Science Foundation
of China under Grant No.12071237.
This work was supported by a grant from  the Postdoctoral Foundation of Zhejiang Province.

\section*{Statements and Declarations}

\subsection*{Funding and/or Conflicts of interests/Competing interests}
\begin{center}
\textbf{Funding and/or Conflicts of interests/Competing interests}
\end{center}
The authors do not have any possible conflicts of interest.

\subsection*{Data Availability Statement}
\begin{center}
\textbf{Data Availability Statement}
\end{center}
The data that support the findings of this study are available from the corresponding author, [BaoShou], upon reasonable request.

\appendix
\section{ Second Fundamental Representation of  $G_2$ }\label{ap}
\begin{figure}[!ht]
  \begin{center}
    \includegraphics[width=6in]{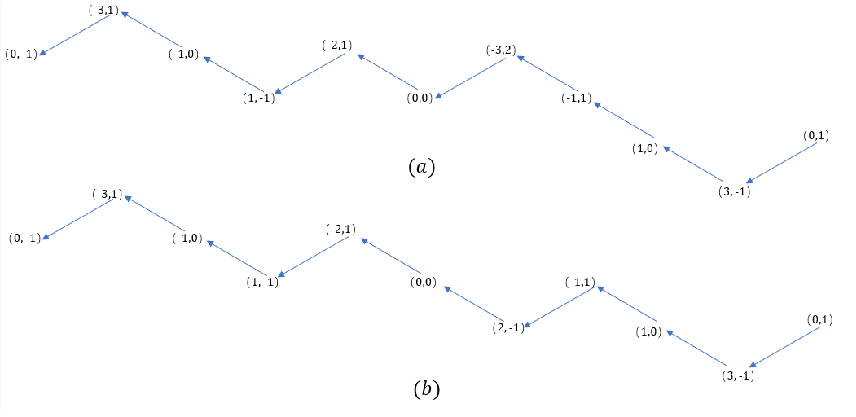}
  \end{center}
  \caption{Paths $a$ and $b$ from the highest-weight state to the lowest weight state in the fundamental representation $(0,1)$ of $G_2$. }
  \label{gab}
\end{figure}
\begin{figure}[!ht]
  \begin{center}
    \includegraphics[width=6in]{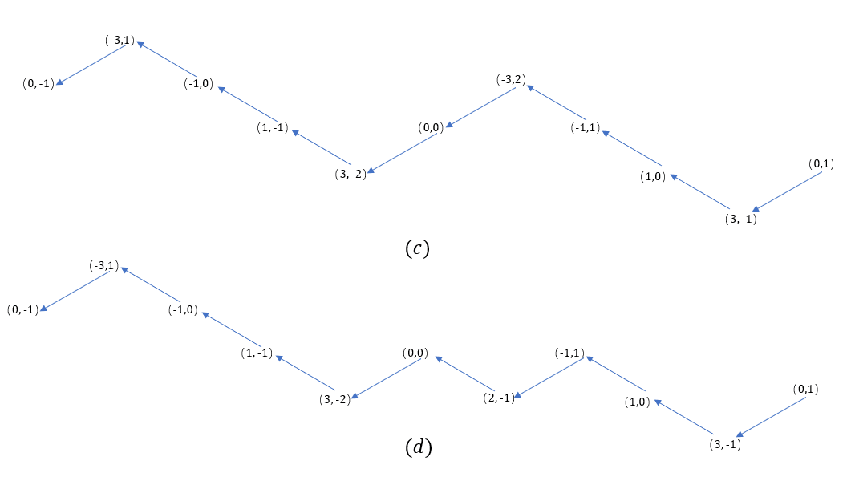}
  \end{center}
  \caption{Paths $c$ and $d$ from the highest-weight state to the lowest weight state in the fundamental representation $(0,1)$ of $G_2$.}
  \label{gcd}
\end{figure}
In this Appendix, using Algorithm \ref{algo}, we calculate  several  inner products  of states corresponding to the weights  of the fundamental representation $(0,1)$ of $G_2$ as shown in Fig.(\ref{G2}).
There are four paths from the  highest-weight vector to the lowest weight vector as shown in Figs.(\ref{gab}) and (\ref{gcd}), corresponding to the  four states $|\upsilon^a_{(0,1)}\rangle$, $|\upsilon^b_{(0,1)}\rangle$, $|\upsilon^c_{(0,1)}\rangle$, and $|\upsilon^d_{(0,1)}\rangle$.

%\begin{eqnarray*}
%  \langle\upsilon^a_{11}|\upsilon^b_{11}\rangle &=&\langle \Lambda_2| E^{+}_{2}E^{+}_{1}E^{+}_{1}E^{+}_{1}E^{+}_{2}|E^{-}_{2}E^{-}_{1}E^{-}_{1}E^{-}_{1}E^{-}_{2}|\Lambda_2\rangle=72 \\
% \langle\upsilon^2_{(0,0)} |\upsilon^1_{(0,0)}\rangle &=&\langle \Lambda_2| E^{+}_{2}E^{+}_{1}E^{+}_{1}E^{+}_{2}E^{+}_{1}| E^{-}_{2}E^{-}_{1}E^{-}_{1}E^{-}_{1}E^{-}_{2}|\Lambda_2\rangle=24 \label{a2a2a2}\\
% \langle\upsilon^2_{(0,0)} |\upsilon^2_{(0,0)}\rangle &=&\langle \Lambda_2|E^{+}_{2}E^{+}_{1}E^{+}_{1}E^{+}_{2}E^{+}_{1}|E^{-}_{1}E^{-}_{2}E^{-}_{1}E^{-}_{1}E^{-}_{2}|\Lambda_2\rangle=36 .
%\end{eqnarray*}

Set the normalization of the highest-weight  state to be one.
$$ \langle\upsilon_{(0,1)} |\upsilon_{(0,1)}\rangle =\langle \Lambda_2|\Lambda_2\rangle=1.$$

The inner product of state $|\upsilon_{(3,-1)}\rangle$ is
$$ \langle\upsilon_{(3,-1)} |\upsilon_{(3,-1)}\rangle =\langle \Lambda_2| E^{+}_{2}E^{-}_{2}|\Lambda_2\rangle=1.$$

The inner product of state $|\upsilon_{(1,0)}\rangle$ is
$$ \langle\upsilon_{(1,0)} |\upsilon_{(1,0)}\rangle=\langle \Lambda_2| E^{+}_{2}E^{+}_{1}E^{-}_{1}E^{-}_{2}|\Lambda_2\rangle=\langle\upsilon_{(3,-1)} |E^{+}_{1}E^{-}_{1}|\upsilon_{(3,-1)}\rangle=3\langle\upsilon_{(3,-1)}|\upsilon_{(3,-1)}\rangle=3.$$

The inner product of state $|\upsilon_{(-1,1)}\rangle$ is
\begin{eqnarray*}
  \langle\upsilon_{(-1,1)} |\upsilon_{(-1,1)}\rangle &=& \langle \Lambda_2| E^{+}_{2}E^{+}_{1}E^{+}_{1}E^{-}_{1}E^{-}_{1}E^{-}_{2}|\Lambda_2\rangle\\
&=&\langle\upsilon_{(3,-1)} |E^{+}_{1}4 E^{-}_{1}|\upsilon_{(3,-1)}\rangle\\
&=&4\langle\upsilon_{(1,0)} |\upsilon_{(1,0)}\rangle\\
&=&12.
\end{eqnarray*}

The inner product of state $|\upsilon_{(-3,2)}\rangle $ is
\begin{eqnarray*}
 \langle\upsilon_{(-3,2)} |\upsilon_{(-3,2)}\rangle
 &=&\langle \Lambda_2| E^{+}_{2}E^{+}_{1}E^{+}_{1}E^{+}_{1}|E^{-}_{1}E^{-}_{1}E^{-}_{1}E^{-}_{2}|\Lambda_2\rangle\\
 &=& \langle\upsilon_{(3,-1)}|E^{+}_{1}E^{+}_{1}3E^{-}_{1}E^{-}_{1}|\upsilon_{(3,-1)}\rangle \\
 &=& 3\langle\upsilon_{(-1,1)} |\upsilon_{(-1,1)}\rangle\\
 &=&36.
\end{eqnarray*}

The inner product of state $|\upsilon_{(2,-1)}\rangle $ is
\begin{eqnarray*}
\langle\upsilon_{(2,-1)} |\upsilon_{(2,-1)}\rangle
 &=&\langle \Lambda_2| E^{+}_{2}E^{+}_{1}E^{+}_{1}E^{+}_{2}| E^{-}_{2}E^{-}_{1}E^{-}_{1}E^{-}_{2}|\Lambda_2\rangle\\
 &=&\langle\upsilon_{(-1,1)} 1 \upsilon_{(-1,1)}\rangle\\
&=&12.
\end{eqnarray*}

The inner product of state $ |\upsilon^1_{(0,0)}\rangle $ is
\begin{eqnarray*}
\langle\upsilon^1_{(0,0)} |\upsilon^1_{(0,0)}\rangle
&=&\langle \Lambda_2| E^{+}_{2}E^{+}_{1}E^{+}_{1}E^{+}_{2}E^{+}_{1}|E^{-}_{1} E^{-}_{2}E^{-}_{1}E^{-}_{1}E^{-}_{2}|\Lambda_2\rangle\\
&=&\langle\upsilon_{(2,-1)} 2\upsilon_{(2,-1)}\rangle\\
&=&24.
\end{eqnarray*}

The inner product  $\langle\upsilon^2_{(0,0)} |\upsilon^1_{(0,0)}\rangle$ is
\begin{eqnarray}\label{p21}
&=&\langle \Lambda_2| E^{+}_{2}E^{+}_{1}E^{+}_{1}E^{+}_{1}E^{+}_{2}|E^{-}_{1} E^{-}_{2}E^{-}_{1}E^{-}_{1}E^{-}_{2}|\Lambda_2\rangle\\
&=&\langle \Lambda_2| E^{+}_{2}E^{+}_{1}E^{+}_{1}E^{+}_{1}|E^{-}_{1} (H_2+E^{-}_{2}E^{+}_{2})|\upsilon_{(-1,1)}\rangle\nonumber\\
&=& \langle\upsilon_{(-3,2)} |\upsilon_{(-3,2)}\rangle\nonumber\\
&=&36.\nonumber
\end{eqnarray}

The inner product of state $|\upsilon^2_{(0,0)}\rangle $ is
\begin{eqnarray*}
\langle\upsilon^2_{(0,0)} |\upsilon^2_{(0,0)}\rangle
&=&\langle \Lambda_2| E^{+}_{2}E^{+}_{1}E^{+}_{1}E^{+}_{1}E^{+}_{2}|E^{-}_{2}E^{-}_{1} E^{-}_{1}E^{-}_{1}E^{-}_{2}|\Lambda_2\rangle\\
&=& \langle\upsilon_{(-3,2)}2\upsilon_{(-3,2)}\rangle\\
&=&72.
\end{eqnarray*}

The inner product  $\langle\upsilon^1_{(0,0)} |\upsilon^2_{(0,0)}\rangle$ is
\begin{eqnarray*}
&=&\langle \Lambda_2| E^{+}_{2}E^{+}_{1}E^{+}_{1}E^{+}_{2}E^{+}_{1}|E^{-}_{2}E^{-}_{1} E^{-}_{1}E^{-}_{1}E^{-}_{2}|\Lambda_2\rangle\\
 &=&\langle \Lambda_2| E^{+}_{2}E^{+}_{1}E^{+}_{2}E^{+}_{1}E^{-}_{2}E^{+}_{1} E^{-}_{1} E^{-}_{1}E^{-}_{1}|\upsilon_{(3,-1)}\rangle\\
 &=&3 \langle\upsilon_{(2,-1)} |\upsilon_{(2,-1)}\rangle\\
&=&36.
\end{eqnarray*}

Next we compute the inner product  $\langle\upsilon^a_{(0,-1)}|\upsilon^a_{(0,-1)}\rangle $ which is the most complicated one for the inner product in the representation
\begin{eqnarray*}
&&\langle\upsilon^a_{(0,-1)}|\upsilon^a_{(0,-1)}\rangle\\
&=&\langle \Lambda_2| E^{+}_{2}E^{+}_{1}E^{+}_{1}E^{+}_{1}E^{+}_{2}E^{+}_{1}E^{+}_{2}E^{+}_{1}E^{+}_{1} (E^{+}_{2} E^{-}_{2}|_{|\upsilon^a_{(-3,1)}}E^{-}_{1}E^{-}_{1}E^{-}_{2}E^{-}_{1}E^{-}_{2}E^{-}_{1}E^{-}_{1}E^{-}_{1}E^{-}_{2}|\Lambda_2\rangle)\\
&=&\langle \Lambda_2| E^{+}_{2}E^{+}_{1}E^{+}_{1}E^{+}_{1}E^{+}_{2}E^{+}_{1}E^{+}_{2}E^{+}_{1}(E^{+}_{1}E^{-}_{1}E^{-}_{1}|_{|\upsilon^a_{(1,-1)}}E^{-}_{2}E^{-}_{1}|_{|\upsilon^a_{(0,0)}}E^{-}_{2}E^{-}_{1}E^{-}_{1}E^{-}_{1}|_{|\upsilon^a_{(3,-1)}}E^{-}_{2}|\Lambda_2\rangle)\\
&=&\langle \Lambda_2| E^{+}_{2}E^{+}_{1}E^{+}_{1}E^{+}_{1}E^{+}_{2}E^{+}_{1}E^{+}_{2}{}_{\langle\upsilon^a_{(1,-1)}|}|E^{+}_{1}\\
&&(0|_{|\upsilon^a_{(-1,0)}}E^{-}_{1}|_{|\upsilon^a_{(1,-1)}}E^{-}_{2}E^{-}_{1}|_{|\upsilon^a_{(0,0)}}E^{-}_{2}E^{-}_{1}E^{-}_{1}E^{-}_{1}|_{|\upsilon^a_{(3,-1)}}E^{-}_{2}|\Lambda_2\rangle\\
&&+0|_{|\upsilon^c_{(-1,0)}}E^{-}_{1}E^{-}_{1}|_{|\upsilon^c_{(3,-2)}}E^{-}_{2}E^{-}_{2}E^{-}_{1}E^{-}_{1}E^{-}_{1}|_{|\upsilon^a_{(3,-1)}}E^{-}_{2}|\Lambda_2\rangle\\
&&+3|_{|\upsilon^d_{(-1,0)}}    E^{-}_{1}E^{-}_{1}    |_{|\upsilon^d_{(3,-2)}}E^{-}_{2}E^{-}_{1}E^{-}_{2}E^{-}_{1}E^{-}_{1}E^{-}_{2}|\Lambda_2\rangle.)
\end{eqnarray*}
The first zero is the coefficient $(\Lambda_{i_0}^0-(n_0-1))n_0=(1-(2-1))2=0$. And the second zero is the eigenvalue of the operators $H_1$. So the first two terms vanish and the last one  is calculated as follows,
\begin{eqnarray*}
&&\langle\upsilon^a_{(0,-1)}|\upsilon^a_{(0,-1)}\rangle\\
&=&3\langle \Lambda_2| E^{+}_{2}E^{+}_{1}E^{+}_{1}E^{+}_{1}E^{+}_{2}E^{+}_{1}E^{+}_{2}{}_{\langle\upsilon^a_{(1,-1)}|}|(E^{+}_{1}
|_{|\upsilon^d_{(-1,0)}}    E^{-}_{1}E^{-}_{1}    |_{|\upsilon^d_{(3,-2)}}E^{-}_{2}E^{-}_{1}E^{-}_{2}E^{-}_{1}E^{-}_{1}E^{-}_{2}|\Lambda_2\rangle)\\
&=&3\langle \Lambda_2| E^{+}_{2}E^{+}_{1}E^{+}_{1}E^{+}_{1}E^{+}_{2}E^{+}_{1}E^{+}_{2}{}_{\langle\upsilon^a_{(1,-1)}|}|(
 4 |_{|\upsilon^d_{(1,-1)}} E^{-}_{1}   E^{-}_{2}|_{|\upsilon^d_{(0,0)}}E^{-}_{1}E^{-}_{2}|_{|\upsilon^d_{(-1,1)}}E^{-}_{1}E^{-}_{1}E^{-}_{2}|\Lambda_2\rangle)\\
&=&12\langle \Lambda_2| E^{+}_{2}E^{+}_{1}E^{+}_{1}E^{+}_{1}E^{+}_{2}E^{+}_{1}{}_{\langle\upsilon^a_{(-2,1)}|}|(
  |_{|\upsilon^d_{(-2,1)}} E^{-}_{1}|_{|\upsilon^a_{(0,0)}}   E^{-}_{2}E^{-}_{1}E^{-}_{1}E^{-}_{1}|_{|\upsilon^a_{(3,-1)}}E^{-}_{2}|\Lambda_2\rangle)\\
&=&36\langle \Lambda_2| E^{+}_{2}E^{+}_{1}E^{+}_{1}E^{+}_{1}E^{+}_{2}{}_{\langle\upsilon^a_{(0,0)}|}|(
 |_{|\upsilon^b_{(0,0)}} E^{-}_{1}   E^{-}_{2}E^{-}_{1}E^{-}_{1}E^{-}_{2}|\Lambda_2\rangle)\\
&=&12*36,
\end{eqnarray*}
where we have used formula (\ref{p21}).

Since the algorithm is an iterative process, we can use the results of the level $k$ when we do the calculations for the step for level $k+1$. For each step, the coefficients    are positive, which are consistent with Conjecture \ref{co}.

%The second term  $II$ is calculated as follows.
%\begin{eqnarray*}
%% \nonumber to remove numbering (before each equation)
%II
%&=&\langle \Lambda_2| E^{+}_{2}E^{+}_{1}E^{+}_{1}E^{+}_{1}E^{+}_{2}E^{+}_{1}E^{+}_{2}{}_{\langle\upsilon^a_{(1,-1)}|}|(E^{+}_{1}
%|_{|\upsilon^c_{(-1,0)}}E^{-}_{1}E^{-}_{1}|_{|\upsilon^c_{(3,-2)}}E^{-}_{2}E^{-}_{2}E^{-}_{1}E^{-}_{1}E^{-}_{1}E^{-}_{2}|\Lambda_2\rangle)\\
%&=&\langle \Lambda_2| E^{+}_{2}E^{+}_{1}E^{+}_{1}E^{+}_{1}E^{+}_{2}E^{+}_{1}E^{+}_{2}{}_{\langle\upsilon^a_{(1,-1)}|}(
%4|_{|\upsilon^c_{(1,-1)}}E^{-}_{1}E^{-}_{2}E^{-}_{2}|_{|\upsilon^c_{(-3,2)}}E^{-}_{1}E^{-}_{1}E^{-}_{1}E^{-}_{2}|\Lambda_2\rangle)\\
%&=&8\langle \Lambda_2| E^{+}_{2}E^{+}_{1}E^{+}_{1}E^{+}_{1}E^{+}_{2}E^{+}_{1}{}_{\langle\upsilon^a_{(-2,1)}|}(
%|_{|\upsilon^a_{(-2,1)}}E^{-}_{1}|_{|\upsilon^a_{(0,0)}}E^{-}_{2}E^{-}_{1}E^{-}_{1}E^{-}_{1}|_{|\upsilon^a_{(3,-1)}}E^{-}_{2}|\Lambda_2\rangle)\\
%&=&8\langle \Lambda_2| E^{+}_{2}E^{+}_{1}E^{+}_{1}E^{+}_{1}E^{+}_{2}{}_{\langle\upsilon^a_{(0,0)}|}(|_{|\upsilon^a_{(0,0)}}E^{-}_{2}|_{|\upsilon^a_{(-3,2)}}E^{-}_{1}E^{-}_{1}E^{-}_{1}E^{-}_{2}|\Lambda_2\rangle\\
%&&+|_{|\upsilon^b_{(0,0)}}E^{-}_{1}E^{-}_{2}|_{|\upsilon^b_{(-1,1)}}E^{-}_{1}E^{-}_{1}E^{-}_{2}|\Lambda_2\rangle)\\
%&=&=8*(24+36)
%\end{eqnarray*}
%
%Combing the terms $I$ and $II$, the inner product is
%$$\langle\upsilon^a_{(0,-1)}|\upsilon^a_{(0,-1)}\rangle=I+II=912. $$

\end{document}